\documentclass{elsarticle}[accept]

\usepackage{lineno,hyperref}

\usepackage{latexsym}
\usepackage{epsfig, ecltree, epic, eepic}
\usepackage{enumerate,amsmath,amssymb,dsfont,pifont,amsthm}
\usepackage[dvips]{color}
\usepackage{graphicx}
\usepackage{multirow}
\usepackage{pdflscape}
\usepackage[table]{xcolor}
\usepackage{subfigure}
\usepackage{xcolor}
\usepackage{relsize, tikz}
\usepackage{etex}

\usetikzlibrary{matrix}

\usepackage[usenames,dvipsnames]{pstricks}
\usepackage{epsfig}
\usepackage{pst-grad} 
\usepackage{pst-plot} 

\usepackage[arrow, matrix, curve]{xy}

\modulolinenumbers[5]

\newcommand{\bbr}{\mathbb{R}}

\newcommand{\bbp}{\mathbb{P}}

\newcommand{\bbn}{\mathbb{N}}

\newcommand{\norm}[1]{\left\| #1 \right\|}

\DeclareMathOperator{\Var}{Var}
\DeclareMathOperator{\Cov}{Cov}

\definecolor{flood1}{RGB}{169,208,142}
\definecolor{flood2}{RGB}{255,242,204}
\definecolor{flood3}{RGB}{255,153,204}

\newcommand{\blind}{0}

\theoremstyle{plain}
\newtheorem{theorem}{Theorem}[section]

\newtheorem{proposition}[theorem]{Proposition}

\theoremstyle{definition}
\newtheorem{definition}[theorem]{Definition}

\theoremstyle{remark}
\newtheorem{remark}[theorem]{Remark}

\journal{Computational Statistics \& Data Analysis}

\nolinenumbers
\begin{document}

\begin{frontmatter}

\title{Generalized Ordinal Patterns Allowing for Ties and Their Applications in Hydrology }

\author[mymainaddress]{Alexander Schnurr}
\ead{schnurr@mathematik-uni-siegen.de}

\author[mysecondaryaddress]{Svenja Fischer\corref{mycorrespondingauthor}}
\cortext[mycorrespondingauthor]{Corresponding author}
\ead{svenja.fischer@rub.de}

\address[mymainaddress]{Department Mathematik, Universit\"at Siegen, D-57072 Siegen, Germany}
\address[mysecondaryaddress]{SPATE research unit, Ruhr-University Bochum, D-44801 Bochum, Germany. Phone: +49 234 - 32 28961}

\begin{abstract}
When using ordinal patterns, which describe the ordinal structure within a data vector, the problem of ties appeared permanently. 
So far, model classes were used which do not allow for ties; randomization has been another attempt to overcome this problem. Often, time periods with constant values even have been counted as times of monotone increase. 
To overcome this, a new approach is proposed: it explicitly allows for ties and, hence, considers more patterns than before.    
Ties are no longer seen as nuisance, but to carry valuable information.    
Limit theorems in the new framework are provided, both, for a single time series and for the dependence between two time series. 

The methods are used on hydrological data sets. It is common to distinguish five flood classes (plus `absence of flood'). Considering data vectors of these classes at a certain gauge in a river basin, one will usually encounter several ties. Co-monotonic behavior between the data sets of two gauges (increasing, constant, decreasing) can be detected by the method as well as spatial patterns. Thus, it helps to analyze the strength of dependence between different gauges in an intuitive way. This knowledge can be used to asses risk and to plan future construction projects.
\end{abstract}


\end{frontmatter}

\section{Introduction}
\label{sec:intro}

\if1\blind{In our paper \cite{fis-sch-sch17}, we }\fi
\if0\blind{In the paper \cite{fis-sch-sch17}, the authors }\fi
have suggested to use so called ordinal patterns for the analysis of hydrological time series. The origin of the concept of ordinal pattern lies in the theory of dynamical systems \cite[cf.][]{bandt05,ban-shi07}. Ordinal patterns have been used in order to analyze the entropy of data sets \cite{ban-pom02} and to estimate the Hurst-parameter of long-range dependent time series \cite[cf.][]{kel-sin10,LRD20}. Additionally, these methods have proved to be useful in the context of EEG-data in medicine \cite[cf.][]{kel-una-una14}, index data in finance \cite[cf.][]{schnurr2014} and flood data in extreme value theory \cite[cf.][]{oestingschnurr2020}. 

The classical approach in \emph{ordinal pattern analysis} works as follows: One decides for a small number $n\in\bbn$, that is, for the length of the data windows under consideration. In each window, only the ordinal information is considered. There are various ways to encode the ordinal information of an $n$-dimensional vector in a pattern. The one that is used most often reads as follows: 

Let $S_{n}$ denote the set of permutations of $\{1, \ldots, n\}$, which we write as $n$-tuples containing each of the numbers $1, \ldots, n$ exactly one time.
By the \emph{ordinal pattern of order $n$} of the vector ${\bf x}=(x^1,...,x^n)$ we refer to the permutation
\begin{align*}
	\Pi(x^1, \ldots, x^n):=(\pi^1,\ldots, \pi^n)\in S_n
\end{align*}
which satisfies
\begin{align*}
	x^{\pi^1}\geq \ldots\geq x^{\pi^n}.
\end{align*}

When using this definition, it is often assumed that the probability of coincident values within the vector $(x^1,\ldots, x^n)$ is zero. 


However, in high-frequency data (mathematical finance) and in categorical data (e.g. hydrology) one will encounter ties between the values. The following approaches have been used in order to overcome this problem in estimation procedures: \\
a) skip the respective data points (leave out these vectors) \\
b) randomize, e.g. by adding a small noise to the data, in order to avoid ties \\
c) add the following to the definition above: 
\begin{center}
	...and $\pi^{j-1}>\pi^j$ if $x^{\pi^{j-1}}=x^{\pi^j}$ for $j\in\{2,\ldots, n\}$.
\end{center}
Each of these `solutions'  has its own drawbacks. In case a) we might loose a lot of information. Considering categorial data sets with a small number of categories, we could even loose most of the information. In case b) we might underestimate co-movement of two data sets (see below) and in case c) the vectors $(4,4,4,4)$ and $(1,10,100,1000)$ are mapped on the same pattern $(4,3,2,1)$, hence, they are considered to be equal from the ordinal point-of-view. 

In earlier papers on the subject, data sets were used which only contained very few vectors with ties. Having categorial data in mind, one could consider cases where the length of patterns exceeds the number of categories. In this case (and there are indeed applications, where this approach is canonical, see below), every single vector has to have ties(!). 

Here, we present a new approach on how to overcome the problem of ties in a natural way. We \emph{explicitly allow for ties and assign a larger number of patterns with the data}. Ties are not seen as being disruptive; they are rather considered to carry valuable information. 

There are various ways on how to describe the patterns one finds in the classical approach. We have already recalled the one which is used most often. However, a different one is much more suitable to be used in the generalized form we have in mind. 

This works as follows: 

Consider a vector of $n$ consecutive values $(x^1,...,x^n)$. 
Write down the values, which are attained $(y^1,...,y^m)$ such that $y^1<y^2<...<y^m$. 
The number of different values attained in $x^1$, ..., $x^n$ is $m\in\bbn$. 
Now define the generalized pattern of $(x^1,...,x^n)$ to be the vector ${\bf t}=(t^1,...,t^n) \in \bbn^n$ such that
\[
t^j=k \text{ if and only if } x^j=y^k.
\] 

By this definition, we obtain that the vector $(1,2,4,3)$ yields the pattern ${\bf t}=(1,2,4,3)$ which is very convenient. Using strictly monotone transformations $\bbn\to\bbr$ one obtains all elements in $\bbr^4$ having this pattern. We write $\Psi$ for the function which assigns the generalized pattern with each vector in contrast to $\Pi$ which has been used in the classical case. The set of all generalized patterns of order $n$ is denoted by $T_n$ (in contrast to the classical permutations $S_n$). 

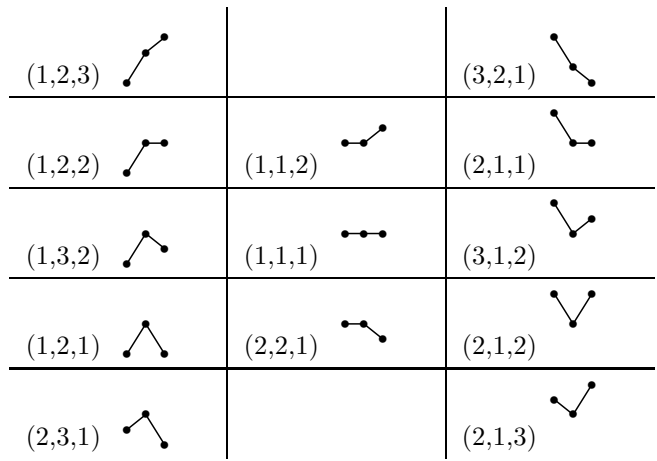
\begin{figure}[h!] \label{13patterns}
	\begin{align*}
		\begin{tabular}{c|c|c} 
			(1,2,3) \setlength{\unitlength}{1cm}
			\scalebox{1} 
			{
				\begin{pspicture}(0,0)(1,1)
					\psline[linewidth=0.02](0,0)(0.25,0.4)(0.5,0.6)
					\psdots[dotsize=0.1](0,0)(0.25,0.4)(0.5,0.6)
				\end{pspicture} 
			}
			& 
			& (3,2,1)
			\scalebox{1} 
			{
				\begin{pspicture}(0,0)(1,1)
					\psline[linewidth=0.02](0,0.6)(0.25,0.2)(0.5,0.0)
					\psdots[dotsize=0.1](0,0.6)(0.25,0.2)(0.5,0.0)
				\end{pspicture} 
			} \\ \hline
			(1,2,2) \setlength{\unitlength}{1cm}
			\scalebox{1} 
			{
				\begin{pspicture}(0,0)(1,1) 
					\psline[linewidth=0.02](0,0)(0.25,0.4)(0.5,0.4)
					\psdots[dotsize=0.1](0,0)(0.25,0.4)(0.5,0.4)
				\end{pspicture} 
			}
			& (1,1,2) \setlength{\unitlength}{1cm}
			\scalebox{1} 
			{
				\begin{pspicture}(0,0)(1,1)
					\psline[linewidth=0.02](0,0.4)(0.25,0.4)(0.5,0.6)
					\psdots[dotsize=0.1](0,0.4)(0.25,0.4)(0.5,0.6)
				\end{pspicture} 
			}
			& (2,1,1)
			\scalebox{1} 
			{
				\begin{pspicture}(0,0)(1,1)
					\psline[linewidth=0.02](0,0.8)(0.25,0.4)(0.5,0.4)
					\psdots[dotsize=0.1](0,0.8)(0.25,0.4)(0.5,0.4)
				\end{pspicture} 
			} \\ \hline
			(1,3,2) \setlength{\unitlength}{1cm}
			\scalebox{1} 
			{
				\begin{pspicture}(0,0)(1,1) 
					\psline[linewidth=0.02](0,0)(0.25,0.4)(0.5,0.2)
					\psdots[dotsize=0.1](0,0)(0.25,0.4)(0.5,0.2)
				\end{pspicture} 
			}
			& (1,1,1) \setlength{\unitlength}{1cm}
			\scalebox{1} 
			{
				\begin{pspicture}(0,0)(1,1)
					\psline[linewidth=0.02](0,0.4)(0.25,0.4)(0.5,0.4)
					\psdots[dotsize=0.1](0,0.4)(0.25,0.4)(0.5,0.4)
				\end{pspicture} 
			}
			& (3,1,2)
			\scalebox{1} 
			{
				\begin{pspicture}(0,0)(1,1)
					\psline[linewidth=0.02](0,0.8)(0.25,0.4)(0.5,0.6)
					\psdots[dotsize=0.1](0,0.8)(0.25,0.4)(0.5,0.6)
				\end{pspicture} 
			} \\ \hline
			(1,2,1) \setlength{\unitlength}{1cm}
			\scalebox{1} 
			{
				\begin{pspicture}(0,0)(1,1) 
					\psline[linewidth=0.02](0,0)(0.25,0.4)(0.5,0.0)
					\psdots[dotsize=0.1](0,0)(0.25,0.4)(0.5,0.0)
				\end{pspicture} 
			}
			& (2,2,1) \setlength{\unitlength}{1cm}
			\scalebox{1} 
			{
				\begin{pspicture}(0,0)(1,1)
					\psline[linewidth=0.02](0,0.4)(0.25,0.4)(0.5,0.2)
					\psdots[dotsize=0.1](0,0.4)(0.25,0.4)(0.5,0.2)
				\end{pspicture} 
			}
			& (2,1,2)
			\scalebox{1} 
			{
				\begin{pspicture}(0,0)(1,1)
					\psline[linewidth=0.02](0,0.8)(0.25,0.4)(0.5,0.8)
					\psdots[dotsize=0.1](0,0.8)(0.25,0.4)(0.5,0.8)
				\end{pspicture} 
			} \\ \hline
			(2,3,1) \setlength{\unitlength}{1cm}
			\scalebox{1} 
			{
				\begin{pspicture}(0,0)(1,1) 
					\psline[linewidth=0.02](0,0.2)(0.25,0.4)(0.5,0.0)
					\psdots[dotsize=0.1](0,0.2)(0.25,0.4)(0.5,0.0)
				\end{pspicture} 
			}
			& 
			& (2,1,3)
			\scalebox{1} 
			{
				\begin{pspicture}(0,0)(1,1)
					\psline[linewidth=0.02](0,0.6)(0.25,0.4)(0.5,0.8)
					\psdots[dotsize=0.1](0,0.6)(0.25,0.4)(0.5,0.8)
				\end{pspicture} 
			} \\ 
		\end{tabular}
	\end{align*}
	\caption{The 13 generalized patterns of length 3}
\end{figure}

As in the case of classical ordinal patterns, one can think of the strings (of numbers) as an archetype structure (see Figure \ref{13patterns} for the case $n=3$).

If ties do not appear, the question of how many patterns one has to consider boils down to the classical case: in how many ways can one write down the numbers 1,2, ..., n. It is well known that this number is $n!=n\cdot (n-1)\cdot ...\cdot 1$. For the generalized patterns we obtain 3 for $n=2$, 13 for $n=3$ and 75 for $n=4$. 

Let us consider the number of patterns under consideration for general $n$: this number is equivalent to the number of ways $n$ competitors can rank in a competition if each competitor is identifiable and if we allow for the possibility of ties. This series in $\bbn$ is known under the name \emph{Fubini numbers}. 
In calculating these numbers one first decides how many first places, second places, third places and so on appear (e.g. (1,1,2,2,3)). For each of these possibilities one then calculates the number of re-arrangements of vectors of this kind (e.g. (1,2,1,3,2)). For the readers' convenience we recall the first 7 Fubini numbers in Figure 2.

\begin{figure}[h] \begin{center}
		\begin{tabular}{|l|l|l|l|l|l|l|l} \hline
			1 & 2 & 3 & 4 & 5 & 6 & 7 & \\  \hline
			1 & 3 & 13 & 75 & 541 & 4683 & 47293 &  \\ \hline
		\end{tabular}
		\caption{The first Fubini numbers}
\end{center}\end{figure}

The name comes from the fact that the numbers coincide with the numbers of possibilities on how to calculate $n$-dimensional integrals iteratively using Fubini's theorem. More information on these numbers can be found at oeis.com. While writing this paper, we were informed that the authors of \cite{Unakavova.2013} have also considered patterns with ties. However, they do not prove any limit theorems or use them in order to analyze dependence structures. They do calculate the number of patterns \cite[cf.][Appendix]{Unakavova.2013}. 

Analyzing the patterns within one time series, we can describe periodicities, times of monotonicity and detect instationarities in a convenient way, just to name a few applications. In the next section, we use ordinal patterns in order to describe the dependence between two time series. Classically, this is often done by using correlation. However, our method works even if the models under consideration do not have second moments.


The notation we are using is more or less standard. 
$(\Omega,\mathcal{F},\mathbb{P})$ denotes a probability space in the background. 
We write $\mathbb{E}$ for the expected value w.r.t. $\mathbb{P}$, 
while $\mathbb{N}$ denotes the positive integers starting with one.

The paper structure is as follows: first, the generalized ordinal patterns are incorporated in a novel measure for dependence and respective sample estimators are introduced (Section \ref{Sec:Methods}). For these measures, limit theorems are provided in Section \ref{Sec:LT}. There, we differentiate between two cases: the generalized ordinal pattern dependence between two time series (in this case for short range dependent data) and the spatial consideration, where the pattern is considered as shifting vector over time. With these two cases, the following application of flood classes in a temporal and spatial manner is supported theoretically (Section \ref{sec:application}).

\section{Methodology}\label{Sec:Methods}

In this section, the above-introduced generalization of the classical ordinal patterns is used to construct a new measure for dependence: \emph{the ordinal pattern dependence}. This dependence measure makes use of the generalized ordinal patterns to obtain information of the co-movement of two time series. It can be extended by a metrical approach that weakens the assumption of strict equality of patterns. Finally, sample estimators for the different measures are given. 

\subsection{Generalized Ordinal Pattern Dependence}

By now we have only considered a single time series, respectively a single data set. Additionally, we would like to use our generalized patterns (with ties) in order to analyze the dependence between time series, respectively data sets \cite[cf.][]{schnurr2014,schnurrdehling2016,BDMS}). The main idea is that a co-monotonic behavior between two data sets leads to an increased number of coincident patterns at the same instants in time. 

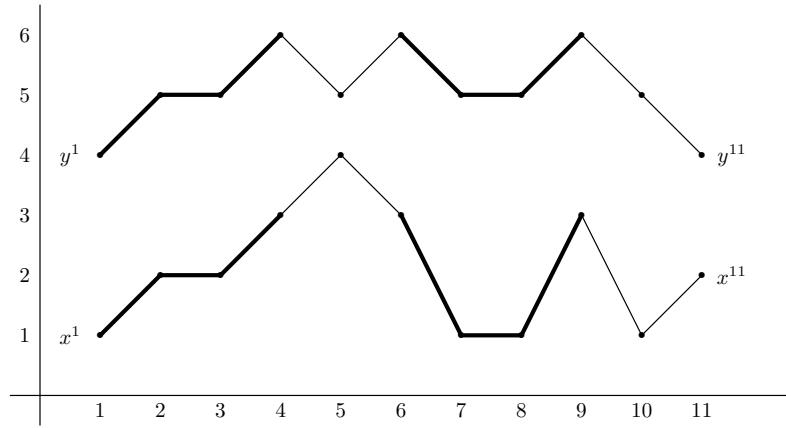
\begin{figure}[h]\begin{center}
		\scalebox{0.8} 
		{
			\begin{pspicture}(0,0)(11,6)
				\psline[linewidth=0.02](-0.5,0)(12.6,0)
				\psline[linewidth=0.02](0,-0.5)(0,6.5)
				\psline[linewidth=0.02](1,4)(2,5)(3,5)(4,6)(5,5)(6,6)(7,5)(8,5)(9,6)(10,5)(11,4)	
				\psline[linewidth=0.02](1,1)(2,2)(3,2)(4,3)(5,4)(6,3)(7,1)(8,1)(9,3)(10,1)(11,2)	
				\psdots[dotsize=0.1](1,4)(2,5)(3,5)(4,6)(5,5)(6,6)(7,5)(8,5)(9,6)(10,5)(11,4)
				\psdots[dotsize=0.1](1,1)(2,2)(3,2)(4,3)(5,4)(6,3)(7,1)(8,1)(9,3)(10,1)(11,2)
				\psline[linewidth=0.07](1,4)(2,5)(3,5)(4,6)
				\psline[linewidth=0.07](6,6)(7,5)(8,5)(9,6)
				\psline[linewidth=0.07](1,1)(2,2)(3,2)(4,3)
				\psline[linewidth=0.07](6,3)(7,1)(8,1)(9,3)
				\rput(1,-0.25){1}
				\rput(2,-0.25){2}
				\rput(3,-0.25){3}
				\rput(4,-0.25){4}
				\rput(5,-0.25){5}
				\rput(6,-0.25){6}
				\rput(7,-0.25){7}
				\rput(8,-0.25){8}
				\rput(9,-0.25){9}
				\rput(10,-0.25){10}
				\rput(11,-0.25){11}
				\rput(-.25,1){1}
				\rput(-.25,2){2}
				\rput(-.25,3){3}
				\rput(-.25,4){4}
				\rput(-.25,5){5}
				\rput(-.25,6){6}	  
				\rput(0.5,1){$x^1$}
				\rput(0.5,4){$y^1$}
				\rput(11.5,2){$x^{11}$}
				\rput(11.5,4){$y^{11}$}
			\end{pspicture} 
		}
		\caption{Two categorial data sets with partially co-monotonic behavior}
\end{center} \end{figure}

Let $X=(X_k)_{k\geq 1}$ and $Y=(Y_k)_{k\geq 1}$ be two stationary time series; let us mention that, although stationarity is \emph{the} classical assumption in this context, our results hold under weaker conditions: stationary increments or even ordinal pattern stationary, that is, stationarity of the pattern probabilities over time, are enough. 

In order to quantify the co-monotonic behavior of two time series, we analyze the probability of coincident patterns at the same instants in time. Since the time series are stationary (or have at least stationary pattern probabilities), it is enough to consider the time points 1 to $n$ 
\[
p_{(X,Y)}:=\bbp\Big(\Psi(X_1, X_2,..., X_n)=\Psi(Y_1,Y_2,...,Y_n)\Big)
\]
and compare this with the hypothetical case of independence 
\[
q_{(X,Y)}:=\sum_{{\bf t}\in T_n} \bbp(\Psi(X_1, X_2,..., X_n)={\bf t})\cdot \bbp(\Psi(Y_1,Y_2,...,Y_n)={\bf t}).
\]
We call $q_{(X,Y)}$ comparison value. If we want to consider an anti-monotonic behavior instead, we consider  $r_{(X,Y)}:=p_{(X,-Y)}$ and $s_{(X,Y)}:=q_{(X,-Y)}$. 
If the time series under consideration are fixed, we omit the subscripts. 
We say that the two time series exhibit generalized positive ordinal pattern dependence if $p>q$; and generalized negative ordinal pattern dependence if $r>s$. 
In analogy to the classical case \cite[][Formula (5)]{schnurrdehling2016} we can define a standardized ordinal pattern coefficient
\begin{align}
	\text{ord}(X,Y):=\left(\frac{p-q}{1-q}\right)^+ - \left(\frac{r-s}{1-s}\right)^+.
\end{align}
This coefficient has been compared in the classical case with other concepts of multivariate dependence in \cite{BDMS}. It admits values between -1 and 1, zero in case of independence and 1 (resp. -1) in case of total positive (resp. negative) dependence.

\subsection{Generalized Average Weighted Ordinal Pattern Dependence}

When dealing with classical ordinal patterns (not allowing for ties), it has proved to be useful to consider not only identical patterns when analyzing co-movement, but to also consider \emph{almost similar patterns}. This generalization of the method is particularly useful, if the patterns are long (big values of $n$) and hence more sensitive to noise. Defining a metric on the patterns (and a weight function) was straight forward in the case of classical patterns without ties. On the space of classical ordinal patterns, that is, permutations,  there exist several metrics which have proved to be useful in different areas of statistics (cf. \cite{dezahuang1998}). When analyzing co-movement, the metric related to the $L_1$-norm has proved to be most suitable. This is due to the fact that permutations are close to each other (measured with this metric) if and only if, high values of one vector correspond to high values of the other vector; and the same holds true for low values. 
However, this metric which was proposed by \cite{schnurrdehling2016}  is not sufficient in our new context. Hence, we have to develop a metric that takes into account the special characteristics of generalized patterns. 

Let us first motivate the new metric: 
Consider the data vectors ${\bf v}_1=(5,5,5,5)$, ${\bf v}_2=(5,5,5,6)$ and ${\bf v}_3=(5,5,5,4)$ of length 4. Intuitively, the differences between the patterns of $\bf{v}_1$ and $\bf{v}_2$ respectively the patterns of $\bf{v}_1$ and $\bf{v}_3$ should coincide. We obtain the patterns: ${\bf t}_1=(1,1,1,1)$, ${\bf t}_2=(1,1,1,2)$ and ${\bf t}_3=(2,2,2,1)$. Applying the $L_1$-metric would lead to different distances: 
\[
\norm{{\bf t}_2-{\bf t}_1}_{1}=1 \text{ and } \norm{{\bf t}_3-{\bf t}_1}_1=3.
\]
In order to overcome this, we consider each pattern as an equivalence class of vectors. Loosely speaking, $(1,1,2,2)$ should be equivalent to $(2,2,3,3)$ and hence, it should be close to $(1,2,3,3)$. Let us emphasize that $(2,2,3,3)$ is not a pattern any more. More precisely, we want to consider not only the pattern itself, but any shift caused by an addition of $k\cdot(1,1,1,1)$ with $k\in \mathbb{Z}$. In the example given here this would mean that we treat the vector ${\bf w}=(2,2,2,2)$ equally to ${\bf t}_1$ such that if we compare ${\bf w}$ and ${\bf t}_3$ instead of ${\bf t}_1$ and ${\bf t}_3$ they have the same distance of 1 as ${\bf t}_1$ and ${\bf t}_2$.

Consequently, we define a new metric for generalized ordinal patterns $d_f$ as follows:

\begin{definition}
Let ${\bf t}, {\bf u}\in T_n$. Then
\begin{align*}
	d_f({\bf t},{\bf u})=\min\limits_{k\in \mathbb{Z}}(\Vert {\bf t}+k\cdot {\bf e}_n- {\bf u}\Vert_{1}), 
\end{align*}
where ${\bf e}_n$ is the unit pattern $(1,\ldots,1)$ of length $n$.
\end{definition}

\begin{proposition}
The mapping $d_f:T_n\times T_n\to \mathbb{R}$ is a metric.
\end{proposition}

\begin{proof}
We have to verify the axioms of a metric. First, consider
\[
d_f({\bf t},{\bf t})=\min\limits_{k\in \mathbb{Z}}(\Vert {\bf t}+k\cdot {\bf e}_n- {\bf t}\Vert_{1})
=\min\limits_{k\in \mathbb{Z}}(\Vert k\cdot {\bf e}_n\Vert_{1})=0
\]
Conversely, if $d_f({\bf t},{\bf u})=0$, there exists a $k\in \mathbb{Z}$ s.t. $\min\limits_{k\in \mathbb{Z}}(\Vert {\bf t}+k\cdot {\bf e}_n- {\bf u}\Vert_{1})=0$. Either $k=0$, then ${\bf t}={\bf u}$ and we are done or $k\neq 0$. In this second case, we obtain the following: ${\bf t}\in T_n$, hence it takes values in $\{1,...,m\}$ and, therefore, ${\bf u}$ takes values in $\{k,...,m+k\} $. This leads to the contradiction ${\bf u}\notin T_n$.

The symmetry property is obtained as follows $({\bf t},{\bf u}\in T_n)$:
\[
d_f({\bf t},{\bf u})
=\min\limits_{k\in \mathbb{Z}}(\Vert {\bf u}-k\cdot {\bf e}_n- {\bf t}\Vert_{1})
=\min\limits_{\ell\in \mathbb{Z}}(\Vert {\bf u}+\ell\cdot {\bf e}_n- {\bf t}\Vert_{1})
=d_f({\bf u},{\bf t})
\]
where we have multiplied the entries of the vector by -1. 
Finally, the triangle inequality follows by $({\bf t},{\bf u},{\bf v}\in T_n)$:
\begin{align*}
d_f({\bf t},{\bf u})+d_f({\bf u},{\bf v})&=\min\limits_{k\in \mathbb{Z}}(\Vert {\bf t}+k\cdot {\bf e}_n- {\bf u}\Vert_{1})+\min\limits_{\ell\in \mathbb{Z}}(\Vert {\bf u}+\ell\cdot {\bf e}_n- {\bf v}\Vert_{1}) \\
&\geq \min\limits_{k\in \mathbb{Z}}\min\limits_{\ell\in \mathbb{Z}} (\Vert {\bf t}+k\cdot {\bf e}_n- {\bf u}+ {\bf u}+\ell\cdot {\bf e}_n- {\bf v}\Vert_{1}) \\
&= \min\limits_{k\in \mathbb{Z}}\min\limits_{\ell\in \mathbb{Z}} (\Vert {\bf t}+(k+\ell)\cdot {\bf e}_n- {\bf v}\Vert_{1})\\
&=\min\limits_{m\in \mathbb{Z}}(\Vert {\bf t}+m\cdot {\bf e}_n- {\bf v}\Vert_{1})
\end{align*}
\end{proof}

In order to implement the new metric, the range of $k$ can be limited to the length of the ordinal pattern such that
\begin{align*}
	d_f({\bf t},{\bf u})=\min\limits_{-n\leq k \leq n}(\Vert {\bf t}+k\cdot {\bf e}_n- {\bf u}\Vert_{1}).
\end{align*}
This simplification still covers all necessary comparisons of patterns.

This metric measures how close patterns are. In case of coincident patterns, we would like to assign the `score' one. The further the patterns are away from each other, the lower the score should be. Following \cite{schnurrdehling2016}, instead of using $d$ directly, we use the following anti-monotone weight function on the values of $d$ for patterns of length less than six 

\begin{align}\label{weightfkt}
	w(x):=1\cdot 1_{\{x=0\} } + 0.5 \cdot 1_{\{x=1\} },
\end{align}

and for patterns of length greater than or equal to six

\begin{align}\label{weightfkt2}
	w(x):=1\cdot 1_{\{x=0\} } + 0.75 \cdot 1_{\{x=1\} } + 0.5 \cdot 1_{\{x=2\} } + 0.25 \cdot 1_{\{x=3\} },
\end{align}

where $1_{\{x=t\}}$ is the indicator function, being one, when $x=t$, and zero elsewhere. Too small steps in the weight function in the presence of small patterns could lead to a too high acceptance of deviations from the observed patterns and hence could increase ordinal pattern dependence artificially. For larger $n$, an appropriate weight function might be defined, for example by using smaller steps of weights. Since we limit the applications to small values of $n$, these weight functions are sufficient for our purposes. Let us remark that the weight functions are different than in the classical case: Since the $L_1$-distance between two permutations in $S_n$ is always an even number, the weight function is - in that case - defined on even numbers only. 

We define the score function for ${\bf t},{\bf u}\in T_n$ as
\begin{align}\label{score}
	s({\bf t},{\bf u}):=w\circ d_f({\bf t},{\bf u}) = w(d_f({\bf t},{\bf u})).
\end{align}

To obtain the score for the dependence between the two time series,  the so-called \emph{total score} is used, defined as
\begin{align*}
	S=\sum_{{\bf t}\in T_n}\sum_{{\bf u}\in T_n}s({\bf t},{\bf u}) \cdot \bbp(\Psi(X_1,...,X_n)={\bf t}, \Psi(Y_1,...,Y_n)={\bf u}).
\end{align*}
The higher the total score is, the stronger the co-movement between the time series (resp. data sets).

Of course, this score can also be applied to the non-weighted case, where simply a sum of 1s is considered. In principle it is possible to calculate a comparison value in the same spirit as above.



An approach via generalized ordinal patterns has (almost) the same advantages as the classical ordinal pattern approach. In addition, we are now able to deal with ties and hence we can handle categorial data. Let us emphasize that the whole analysis is stable under strictly monotone transformations of the state space. Structural breaks in a single time series do not affect the ordinal pattern dependence very much. In hydrology, e.g. a shift in mean can be caused by a new dam. The algorithms in order to analyze data sets by our method are very quick and in the analysis of  `co-monotonic behavior' no second moments are needed. These are inevitable when applying classical correlation. 

Small noise might (in contrast to the classical approach) change the ordinal structures, since ties might become increasing or decreasing patterns. However, keep in mind that the new approach is designed in first place to deal with categorial data, where a small noise does not make any sense.

\subsection{Natural Estimators for the Ordinal Pattern Coefficient and the Total Score}

The above given theoretical characteristics for ordinal patterns can be estimated in the most natural way by their sample analogues. Given the observations of $\left(X_1,\ldots,X_N\right)$ and $\left(Y_1,\ldots,Y_N\right)$, the following estimators can be derived.

For the probability of a coincident ordinal pattern $p$, the sample estimator is given by
\[
\hat{p}_N=\frac{1}{N-n+1} \sum_{j=1}^{N-n+1} 1_{\{ {\Psi(X_j,...,X_{j+n-1})=\Psi(Y_j,...,Y_{j+n-1}) }\}}.
\]

A similar estimator was used in \cite{schnurrdehling2016} in the context of classical ordinal patterns.

In order to estimate the standardized ordinal pattern coefficient, additional parameters are required. Sample estimators for these parameters $q$, $r$ and $s$ are given in the following

\begin{align*}
	\hat{q}_N &=\sum_{{\bf t}\in T_n}\left(\frac{1}{N-n+1}\right)^2\sum_{i=1}^{N-n+1}1_{\{\Psi(X_i,\ldots,X_{i+n-1})={\bf t}\}}\sum_{i=1}^{N-n+1}1_{\{\Psi(Y_i,\ldots,Y_{i+n-1})={\bf t}\}}\\
	\hat{r}_N&=\frac{1}{N-n+1}\sum_{i=1}^{N-n+1}1_{\{\Psi(X_i,\ldots,X_{i+n-1})=\Psi(-Y_{i},\ldots,-Y_{i+n-1})\}}\\
	\hat{s}_N&=\sum_{{\bf t}\in T_n}\left(\frac{1}{N-n+1}\right)^2\sum_{i=1}^{N-n+1}1_{\{\Psi(X_i,\ldots,X_{i+n-1})={\bf t}\}}\sum_{i=1}^{N-n+1}1_{\{\Psi(-Y_i,\ldots,-Y_{i+n-1})={\bf t}\}}.
\end{align*}

With these, the standardized ordinal pattern coefficient can be estimated by

\begin{align*}
	\widehat{ord}(X,Y)=\left(\frac{\hat{p}_N-\hat{q}_N}{1-\hat{q}_N}\right)^+-\left(\frac{\hat{r}_N-\hat{s}_N}{1-\hat{s}_N}\right)^+.
\end{align*}

When considering the total score $S$, a sample estimator is given by

\begin{align*}
	\hat{S}_N=\frac{1}{N-n+1}\sum_{i=1}^{N-n+1}\omega(d_f(\Psi(X_i,\ldots,X_{i+n-1}), \Psi(Y_i,\ldots,Y_{i+n-1}))).
\end{align*}

Here, the considered window for the pattern is shifted in time by one step each time. Alternatively, the window could also be shifted by the length of the ordinal pattern $n$. This way, an overlapping of the windows could be omitted. However, it is not always a priori known, which time steps should be part of one window. If, e.g., monthly data is used, a natural window would be the choice of one year, capturing the seasonality. However, for annual data this is no longer obvious. Hence, a loss of information could be obtained when choosing a step of length $n$.

\section{Limit Theorems}\label{Sec:LT} 

In this section, we provide limit theorems under short range dependence for the previously introduced estimators in the context of ordinal pattern dependence. 
Moreover, a limit theorem for the multivariate case, i.e., when considering a pattern vector that is shifted in time, is given.
These limit theorems prepare the applications in our temporal respectively spatial approach in the subsequent section. 

\subsection{Pairwise comparison in an SRD regime}

The first case that is handled is the ordinal pattern dependence. This means limit theorems are required for a bivariate vector of random variables.
We cannot use the limit theorems as presented in \cite{schnurrdehling2016} for the following reason: 
Already the fundamental first theorem in that paper needs the assumption that the distribution functions of the underlying random variables are Lipschitz continuous. Our random variables might be discrete (cf. the application on flood alert classes in Section \ref{sec:application}) and hence a more generalized setting is required.




Limit theorems are developed under a general setting of short-range dependence to consider auto-correlation in the data. Here, we consider absolute regularity. 

\begin{definition}
	Let $\mathcal{A}, \mathcal{B}\subset\mathcal{F}$ be two $\sigma$-fields on the probability space $(\Omega,\mathcal{F},\mathbb{P})$. The absolute regularity coefficient of $\mathcal{A}$ and $\mathcal{B}$ is given by
	\begin{align*}
		\beta(\mathcal{A},\mathcal{B})=\mathbb{E}\sup\limits_{A\in\mathcal{A}}\big| \mathbb{P}(A\vert \mathcal{B})-\mathbb{P}(A)\big|.
	\end{align*}
	Independence of $\mathcal{A}$ and $\mathcal{B}$ is equivalent to $\beta(\mathcal{A},\mathcal{B})=0$.
	
	If $(X_k)_{k\in\mathbb{N}}$ is a stationary process, then the absolute regularity coefficients of $(X_k)_{k\in\mathbb{N}}$ are given by 
	\begin{align*}
		\beta(\ell)=\sup\limits_{k\in\mathbb{N}}\beta(\mathcal{F}_1^k,\mathcal{F}_{k+\ell}^\infty).
	\end{align*}
	$(X_k)_{k\in\mathbb{N}}$ is called absolutely regular, if $\beta(\ell)\rightarrow0$ as $l\rightarrow \infty$.
\end{definition}

This concept of short range dependence has been introduced by \cite{Volkonskii.1959} (and there attributed to Kolmogorov) and includes many well-known models used in applications like certain Markov chains or AR-processes. Gaussian sequences are also absolute regular, as long as their spectral density is of a certain form \cite{Bradley.2005}. Moreover, renewal processes and solutions of stochastic differential equations were shown to be absolutely regular random processes.

Under this dependency assumption, the following limit theorem holds.

\begin{theorem}
	Let $(Z_k=(X_k,Y_k))_{k\in\mathbb{N}}$ be an absolutely regular, stationary process on the probability space $(\Omega,\mathcal{F},\mathbb{P})$.
	Then,
	\[
	\hat{p}_N \rightarrow p \hspace{5mm} \mathbb{P}\text{-a.s.}
	\]
\end{theorem}

The theorem is in line with \cite{LRD20} and \cite{kel-sin10}, where similar results for the long-range dependent and weak dependence  cases are stated.

\begin{proof}
	Since we assume that the underlying short-range dependent process is absolute regular ($\beta$-mixing), it is also ergodic. This makes the ergodic theorem by Birkhoff applicable. We apply it to the partial sums 
	
	$$\sum_{i=1}^{N-n+1}\left(f(X_i,\ldots,X_{i+n-1},Y_i,\ldots,Y_{i+n-1})-\mathbb{E}f(X_1,\ldots,X_n, Y_1,\ldots,Y_n)\right)$$
	
	with $f(X_i,\ldots,X_{i+n-1},Y_i,\ldots,Y_{i+n-1})=1_{\{\Psi(X_i,\ldots,X_{i+n-1})=\Psi(Y_i,\ldots,Y_{i+n-1})\}}$ to obtain
	
	$$\frac{1}{N-n+1}\sum_{i=1}^{N-n+1}f(X_{i},\ldots,X_{i+n-1},Y_i,\ldots,Y_{i+n-1})\rightarrow \mathbb{E}f(X_1,\ldots,X_n,Y_1,\ldots,Y_n)=p $$
	\hspace{5mm} $\mathbb{P}\text{-almost surely}$.
	
\end{proof}

Besides consistency of the estimator, also a limit distribution is helpful, e.g. to derive confidence bands. The respective theorem is given in the following, together with limit theorems for the remaining parameters $q$, $r$ and $s$, such that also a limit distribution of the standardized ordinal pattern coefficient can be obtained.

\begin{theorem} \label{thm:limitp}
	Let $(Z_k=(X_k,Y_k))_{k\in\mathbb{N}}$ be an absolutely regular, stationary process on the probability space $(\Omega,\mathcal{F},\mathbb{P})$ with coefficients $\sum_{\ell=1}^\infty(\beta(\ell))^{1/2}<\infty$.
	Then,
	\begin{align*}
		\sqrt{N}(\hat{p}_N-p)\rightarrow \mathcal{N}(0,\sigma^2),
	\end{align*}
	where 
	\begin{align*}
		\sigma^2=&\Var(1_{\{\Psi(X_1,\ldots,X_{n})=\Psi(Y_1,\ldots,Y_{n})\}})\\
		& +2\sum_{m=2}^\infty \Cov(1_{\{\Psi(X_1,\ldots,X_{n})=\Psi(Y_1,\ldots,Y_{n})\}},1_{\{\Psi(X_m,\ldots,X_{m+n-1})=\Psi(Y_m,\ldots,Y_{m+n-1})\}}).
	\end{align*}
\end{theorem}

The concept of the bivariate mixing process used here is very similar to the theorem given in \cite{Doukhan.2015}. There, a limit theorem is proven for a strongly mixing process by using the respective inequalities of \cite{Rio.2000}. We will obtain a similar result by applying a theorem of \cite{Ibragimov.1962}, that explicitly omits any continuity assumption on the distribution.
\begin{proof}
	First, define $(f_i)_{i \in \mathbb{N}}$ with $f_i=1_{\{\Psi(X_i,\ldots,X_{i+n-1})=\Psi(Y_i,\ldots,Y_{i+n-1})\}}-p$ as stationary process derived from the random variable $f(Z_1,\ldots, Z_{i-1},Z_i,Z_{i+1},\ldots)$. Though the stationary process may not automatically be absolutely regular, Theorem 2.1 of \cite{Ibragimov.1962} gives a theorem with sufficient conditions to nevertheless derive a limit theorem for $(f_i)_{i \in \mathbb{N}}$.
	It is $\mathbb{E}(f_1)=\mathbb{P}(\Psi(X_1,\ldots,X_{n})=\Psi(Y_1,\ldots,Y_{n}))-p=0$ and $\mathbb{E}f_1^2<\infty$ due to the boundedness of $f_1$ as and indicator function and hence assumption one of Theorem 2.1 of \cite{Ibragimov.1962} follows from stationarity. The remaining assumptions of Theorem 2.1 follow from the assumptions made in our theorem.
	Therefore, we can apply Theorem 2.1 of \cite{Ibragimov.1962} to $(f_i)_{i\in\mathbb{N}}$. This gives us	
	
	\begin{align*}
		\mathbb{P}\left(\frac{f_1+\ldots+f_N}{\sigma\sqrt{N}}<x\right)\longrightarrow\frac{1}{\sqrt{2\pi}}\int_\infty^x e^{-u^2/2}du
	\end{align*}
	for $N\rightarrow\infty$ and thus 
	\begin{align*}
		\frac{1}{\sqrt{N}}\sum_{i=1}^N\Big(\mathbb{E}(\Psi(X_i,\ldots,X_{i+n-1})=\Psi(Y_i,\ldots,Y_{i+n-1}))-p\Big)=\sqrt{N}(\hat{p}_N-p)\rightarrow \mathcal{N}(0,\sigma^2),
	\end{align*}
	for $\sigma\neq 0$.
	
	The variance $\sigma^2$ is given by
	\begin{align*}
		\sigma^2=&\mathbb{E}f_1^2+2\sum_{j=1}^\infty \mathbb{E}(f_1f_i)\\
		=&\Var(1_{\{\Psi(X_1,\ldots,X_{n})=\Psi(Y_1,\ldots,Y_{n})\}})\\
		& +2\sum_{m=2}^\infty \Cov(1_{\{\Psi(X_1,\ldots,X_{n})=\Psi(Y_1,\ldots,Y_{n})\}},1_{\{\Psi(X_m,\ldots,X_{m+n-1})=\Psi(Y_m,\ldots,Y_{m+n-1})\}})<\infty,
	\end{align*}
	which completes the proof.
\end{proof}

\begin{remark}
	For applications, it might be helpful to estimate $\sigma$ in the previous theorem since in general this parameter is unknown. Such an estimator $\hat{\sigma}_N$, that is consistent, can be derived by applying the sample estimators for $U$-statistic variances in \cite{deJong.2000}. Then,
	\begin{align*}
		\hat{\sigma}^2_N=\frac{1}{N}\sum_{i=1}^{N-n+1}\sum_{j=1}^{N-n+1}k\left(\frac{i-j}{b_N}\right)&\left(1_{\{\Psi(X_i,\ldots,X_{i+n-1})=\Psi(Y_i,\ldots,Y_{i+n-1})\}}-\hat{p}_N\right)\\
		&\left(1_{\{\Psi(X_j,\ldots,X_{j+n-1})=\Psi(Y_j,\ldots,Y_{j+n-1})\}}-\hat{p}_N\right).
	\end{align*}

	Application of Slutsky's Theorem gives
	
	\begin{align*}
		\sqrt{N}\frac{\hat{p}_N-p}{\hat{\sigma}_N}\rightarrow \mathcal{N}(0,1).
	\end{align*}
	
	under the conditions of Theorem \ref{thm:limitp}.
\end{remark}

Additionally, limit theorems for the remaining parameters which have been introduced above can be obtained. 

\begin{theorem}
	Under the same conditions as in Theorem \ref{thm:limitp},
	
	\begin{align*}
		\sqrt{N}(\hat{q}_N-q)\rightarrow \mathcal{N}(0,\delta^2),
	\end{align*}
	where 
	\begin{align*}
		\delta^2=\sum_{{\bf t}_1,{\bf t}_2\in T_{n}}&\mathbb{P}(\Psi(Y_1,\ldots,Y_{n})={\bf t}_1)\sum_{k=-\infty}^\infty \Cov(1_{\{\Psi(X_1,\ldots,X_{n})={\bf t}_1\}},1_{\{\Psi(X_k,\ldots,X_{k+n-1}={\bf t}_2)\}})\\ &\mathbb{P}(\Psi(Y_1,\ldots,Y_{n})={\bf t}_2)\\
		+ 2\sum_{{\bf t}_1,{\bf t}_2\in T_{n}}&\mathbb{P}(\Psi(X_1,\ldots,X_{n})={\bf t}_1)\sum_{k=-\infty}^\infty \Cov(1_{\{\Psi(X_1,\ldots,X_{n})={\bf t}_1\}},1_{[\Psi(Y_k,\ldots,Y_{k+n-1}={\bf t}_2)\}})\\
		&\mathbb{P}(\Psi(Y_1,\ldots,Y_{n})={\bf t}_2)\\
		+ \sum_{{\bf t}_1,{\bf t}_2\in T_{n}}&\mathbb{P}(\Psi(X_1,\ldots,X_{n})={\bf t}_1)\sum_{k=-\infty}^\infty \Cov(1_{[\Psi(Y_1,\ldots,Y_{n})={\bf t}_1\}},1_{[\Psi(Y_k,\ldots,Y_{k+n-1}={\bf t}_2)\}})\\
		&\mathbb{P}(\Psi(X_1,\ldots,x_{n})={\bf t}_2)
	\end{align*}
\end{theorem}

\begin{proof}
	We will first treat the two random vectors $(X_k)_{k\in \mathbb{N}}$ and $(Y_k)_{k\in\mathbb{N}}$ separately. Define
	
	\begin{align*}
		q_X(\textbf{t})&=\mathbb{P}(\Psi(X_1,\ldots,X_{n})=\textbf{t})\\
		q_Y(\textbf{t})&=\mathbb{P}(\Psi(Y_1,\ldots,Y_{n})=\textbf{t})
	\end{align*}
	with $\textbf{t}\in T_n$ and the respective sample analogons
	\begin{align*}
		\hat{q}_X(\textbf{t})&=\frac{1}{N}\sum_{i=1}^{N-n+1}1_{\{\Psi(X_i,\ldots,X_{i+n-1})=\textbf{t}\}}\\
		q_Y(\textbf{t})&=\frac{1}{N}\sum_{i=1}^{N-n+1}1_{\{\Psi(Y_i,\ldots,Y_{i+n-1})=\textbf{t}\}}.
	\end{align*}
	Applying again Theorem 2.1 of \cite{Ibragimov.1962}, where the conditions are fulfilled due to the boundedness of the indicator function and the assumptions on the absolute regular coefficients, together with the Cram\'er-Wold device it follows in the multivariate case that 
	
	\begin{align*}
		&\sqrt{N}((\hat{q}_X(\textbf{t})-q_X(\textbf{t}))^T_{\textbf{t}\in T_n}, (\hat{q}_Y(\textbf{t})-q_Y(\textbf{t}))^T_{\textbf{t}\in T_n})^T\rightarrow \mathcal{N}(0,\Sigma)\\
	\end{align*}
	
	with
	
	\begin{align*}
		\Sigma=&\left(\begin{array}{rr} 
			\Sigma_{X} & \Sigma_{XY} \\ 
			\Sigma_{YX} & \Sigma_{Y}  
		\end{array}\right)
	\end{align*}
	
	and 
	\begin{align*}
		\Sigma_{X}&=\left(\sum_{m=1}^\infty\Cov(1_{\{\Psi(X_1,\ldots,X_{n})=\textbf{t}_1\}},1_{\{\Psi(X_m,\ldots,X_{m+n-1})=\textbf{t}_2\}})\right)_{\textbf{t}_1,\textbf{t}_2\in T_n}\\
		\Sigma_{XY}&=\Sigma_{YX}=\left(\sum_{m=1}^\infty\Cov(1_{\{\Psi(X_1,\ldots,X_{n})=\textbf{t}_1\}},1_{\{\Psi(Y_m,\ldots,Y_{m+n-1})=\textbf{t}_2\}})\right)_{\textbf{t}_1,\textbf{t}_2\in T_n}\\
		\Sigma_{Y}&=\left(\sum_{m=1}^\infty\Cov(1_{\{\Psi(Y_1,\ldots,Y_{n})=\textbf{t}_1\}},1_{\{\Psi(Y_m,\ldots,Y_{m+n-1})=\textbf{t}_2\}})\right)_{\textbf{t}_1,\textbf{t}_2\in T_n}.
	\end{align*}	
	
	In the next step, we want to apply the delta method. For this purpose, we define the function $f:\mathbb{N}^{2n!}\rightarrow\mathbb{R}$ with $f(\textbf{x},\textbf{y})=\sum_{i=1}^{n!}x_iy_i$ with $\textbf{x}=(x_1,\ldots,x_{n!})$ and $\textbf{y}=(y_1,\ldots,y_{n!})$ such that $q=f(q_X,q_Y)$ and $\hat{q}_N=f((\hat{q}_X(\textbf{t}))_{\textbf{t}\in T_n},(\hat{q}_Y(\textbf{t}))_{\textbf{t}\in T_n} )$.

	It remains to show that $f$ is differentiable. To show this, we first have a look at the partial derivatives 
	
	\begin{align*}
		\frac{\partial}{\partial x_i}f(\textbf{x},\textbf{y})=y_i \text{ and } \frac{\partial}{\partial y_i}f(\textbf{x},\textbf{y})=x_i.
	\end{align*}
	Hence, $\Delta f=(x,y)=\left(\begin{array}{ll}\textbf{x} \\ \textbf{y} \end{array}\right)$	which make the delta method applicable. This results in
	\begin{align*}
		\sqrt{N}(\hat{q}_N-q)\rightarrow \mathcal{N}(0,\delta^2)
	\end{align*}
	with 
	
	\begin{align*}
		\delta^2=(\Delta f(q_X,q_Y))^T\Sigma \Delta f(q_X,q_Y)=(q_X^T,q_Y^T)\Sigma\left(\begin{array}{ll}\textbf{x} \\ \textbf{y} \end{array}\right)
	\end{align*}
	which completes the proof.
\end{proof}

Analogous results can be obtained for the sample estimators of $r$ and $s$ using the same techniques. With these results, the limit theorem for $\widehat{ord}(X,Y)$ then follows simply from the delta method using the assumptions from Theorem \ref{thm:limitp}.

Secondly, limit theorems for the sample estimator of the total score are required. Under the same conditions as in the above theorems, it can be shown that
\begin{align*}
	\sqrt{N}(\hat{S}_N-S)\rightarrow N(0,\gamma^2)
\end{align*}
with 
\begin{align*}
	\gamma^2=&\Var(\omega(d_f(\Psi(X_1,\ldots,X_{n})=\Psi(Y_1,\ldots,Y_{n}))))\\
	& +2\sum_{m=2}^\infty \Cov(\omega(d_f(\Psi(X_1,\ldots,X_n)=\Psi(Y_1,\ldots,Y_{n}))),\\
	& \omega(d_f(\Psi(X_m,\ldots,X_{m+n-1})=\Psi(Y_m,\ldots,Y_{m+n-1})))).
\end{align*}
by application of the Central Limit Theorem for functionals of absolutely regular processes of \cite{Ibragimov.1962} and the Cram\'er-Wold device.

\subsection{Limit Theorems in a Multivariate i.i.d. Setting}

Dealing with multivariate data, we also consider the patterns at each point in time. Let $(X_k)_{k\in\bbn}$ be a $d$-dimensional time series. The pattern $\Psi((X_k^1,...,X_k^d))$ describes in a short way, where the highest value, second highest value... in the vector $X_k$ is. While we have dependence within the vector $X_k$, we encounter situations where the time series $(X_k)_{k\in\bbn}$ consists of independent random vectors. Hence, we can use classical limit theorems: 
\[
\frac{1}{N} \sum_{k=1}^N 1_{\{\Psi(X_k)=\bf{t}\} } \to \bbp(\Psi(X_1)=\bf{t}) \hspace{10mm} \bbp\text{-a.s.}
\]
holds by the (strong) law of large numbers. Writing $P_t:=\bbp(\Psi(X_1)=\bf{t})$ we obtain in addition
\[
\frac{1}{\sqrt{N}} \sum_{k=1}^{N} (1_{\{\Psi(X_k)=\bf{t}\} }-P_t) \to N(0,P_t^2)
\]
by the Central Limit Theorem for i.i.d. Bernoulli random variables. These simple convergence theorems are stated here only for the sake of completeness since they are used in the spatial approach of the subsequent section. 

\section{Application and Simulation: Temporal and Spatial Patterns in Flood Classification}
\label{sec:application}

Let us consider the following general situation: we observe $d$-dimensional data vectors at discrete time points $k\in\bbn$. This is modeled as a $d$-dimensional time series $(X_k)_{k\in \bbn}$. We want to analyze the dependence between the one-dimensional time series $(X^{a}_k)_{k\in \bbn}$ and $(X^{b}_k)_{k\in \bbn}$ for $1\leq a,b \leq d$. We will do this using two different approaches which both make use of generalized ordinal patterns. 

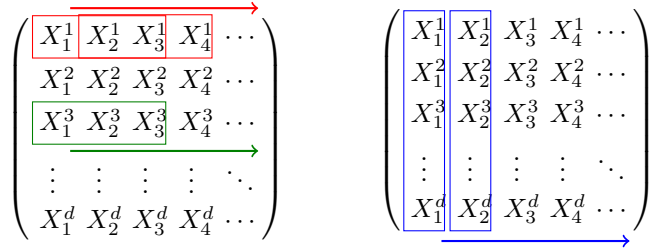
\begin{figure}[h!]
	\begin{center}
		\begin{tikzpicture}
			\def\a{\mathlarger{\mathlarger{\boldsymbol{\ast}}}}
			\def\b{{\bullet}}
			\matrix (m) [matrix of math nodes,
			inner sep=0pt, column sep=0.25em, 
			nodes={inner sep=0.25em,text width=1em,align=center},
			left delimiter=(,right delimiter=),
			]
			{%
				X_1^1 & X_2^1 \textcolor{white}{\big  \vert}& X_3^1 & X_4^1& \cdots  \\
				X_1^2 & X_2^2 & X_3^2 & X_4^2& \cdots \\
				X_1^3 & X_2^3 \textcolor{white}{\big  \vert}& X_3^3 & X_4^3& \cdots \\
				\vdots & \vdots & \vdots & \vdots &  \ddots   \\
				X_1^d & X_2^d & X_3^d & X_4^d& \cdots \\
			};%
			
			
			\draw[red,] (m-1-2.north west) rectangle (m-1-4.south east); 
			\draw[red,] (m-1-1.north west) rectangle (m-1-3.south east); 
			
			\draw[green!50!black] (m-3-1.north west) rectangle (m-3-3.south east); 
			\draw[red,thick,->] (-1,1.6) -- (1.5,1.6);
			\draw[green!50!black,thick,->] (-1,-0.3) -- (1.5,-0.3);
		\end{tikzpicture} \hspace{10mm}
		\begin{tikzpicture}
			\def\a{\mathlarger{\mathlarger{\boldsymbol{\ast}}}}
			\def\b{{\bullet}}
			\matrix (m) [matrix of math nodes,
			inner sep=0pt, column sep=0.25em, 
			nodes={inner sep=0.25em,text width=1em,align=center},
			left delimiter=(,right delimiter=),
			]
			{%
				X_1^1 & X_2^1 & X_3^1 & X_4^1& \cdots  \\
				X_1^2 & X_2^2 & X_3^2 & X_4^2& \cdots \\
				X_1^3 & X_2^3 & X_3^3 & X_4^3& \cdots \\
				\vdots & \vdots & \vdots & \vdots &  \ddots   \\
				X_1^d & X_2^d & X_3^d & X_4^d& \cdots \\
			};%
			
			
			\draw[blue] (m-1-1.north west) rectangle (m-5-1.south east);
			\draw[blue] (m-1-2.north west) rectangle (m-5-2.south east);
			\%draw[green!50!black] (m-3-1.north west) rectangle (m-3-3.south east); 
			\draw[blue,thick,->] (-1,-1.6) -- (1.5,-1.6);
		\end{tikzpicture}
	\end{center}
	\caption{The temporal approach (left) and the spatial approach (right)}
\end{figure}

I. Temporal approach: we compare the time series $(X^{a}_k)_{k\in \bbn}$ (for $1\leq a \leq d$ pairwise). To this end, we use the ordinal pattern dependence as described above. This allows us to analyze co-monotonic behavior over time. 

II. Spatial approach: we assign the patterns $\Psi(X_k)$ with vectors $X_k$ at each fixed time point $k$. If the probability 
\[
\bbp(\Psi(X^a_1)=\Psi(X^b_1))
\]
is high, this is an indicator for strong dependence between the $a$th and the $b$th component of the vector. However, the spatial pattern contains more information. The frequency of certain patterns can reveal structures in a network.
If, for example, the frequency of the unit pattern $(1,1,...,1)$ is high, this implies an overall high dependence between all vector entries. 
In the following, the example of flood classes will be introduced. If a gauging network within a river basin is treated as vector, the application of spatial ordinal patterns can deliver information on which gauges often react similarly, due to weather patterns or anthropological impacts.

Flood events within a river basin almost always occur at several gauges at the same time. Due to spatially distributed rainfall, snowmelt and flood wave propagation, flood events show patterns within a basin. Especially weather patterns like the \emph{Vb} weather pattern, that is responsible for almost all large flood events in Germany over the last decades, are characterized by a certain circulation \cite{nissen2013}. Hence, spatial patterns of flood events are investigated by many researchers to understand the extent and possible hazard of extreme flood events \cite{Uhlemann.2010}. Nevertheless, it is not straightforward to compare the observed magnitude of the events at different gauges. The peak of the flood wave, which is mostly used to characterize a flood event and for flood statistics \cite{Salas.2014}, depends much on the catchment size and characteristics and can be influenced by overlaid waves of tributaries. The same holds true for the flood volume. Instead, \cite{Fischer.2018} propose a flood classification. This classification has the advantage that it is distribution-free and ordinal but, since it is an extension of the Chebychev-inequality, still depends on probabilities, which is necessary for flood statistics. In their research, they also detected certain patterns in the classification of single flood events within the river basin, but they have not been evaluated statistically concerning significance. Hence, we will adopt this classification and investigate the significance of the spatial patterns of flood classes with ordinal patterns. 
For this purpose, according to \cite{Fischer.2018}, five flood classes are defined (Table \ref{Tab:floodclass}).

\begin{table}[!h]
	\centering
	\begin{tabular}{c | c c }
		Flood Class & Description & Non-exceedance Probability $p$ of Peak\\
		\hline
		0 & small flood & $p<0.5$\\
		\cellcolor{flood1} 1 & \cellcolor{flood1} normal flood & $0.5\leq p < 0.8$\\
		\cellcolor{flood2} 2 & \cellcolor{flood2} medium flood & $0.8\leq p \leq 0.933$\\
		\cellcolor{flood3} 3 & \cellcolor{flood3} large flood & $0.933\leq p \leq 0.966$\\
		\cellcolor{red} 4 & \cellcolor{red} very large flood & $0.966\leq p $\\
	\end{tabular}
	\caption{Five flood classes to characterize the magnitude of flood events.}
	\label{Tab:floodclass}	
\end{table}

Of course, there does not exist spatial coherence of flood events only, but also temporal patterns can be observed for floods. In hydrology, often so called flood-rich and flood-poor periods are observed. These are defined as significant clusters of floods above or below a certain threshold. Nevertheless, significance of these clusters is hardly detected with common techniques, especially for non-annual data \cite{Lun.2020}. Again, the use of flood classes instead of thresholds in combination with ordinal patterns shall help us to detect such significant patterns in the temporal series.

For our analyses, 14 gauges of the Mulde river basin with an observation period starting in 1926 (due to uniformity) are available, including head as well as estuary catchments. For these gauges, 314 flood events have been separated and classified with the classification described above. As a reference gauge we use the Golzern gauge, being the one at the estuary and hence including the whole basin. For this gauge, the classification resulted in 9 flood events of class 4 (including the extraordinary large floods in 2002 and 2013), 7 flood events of class 3, 39 events of class 2, 56 events of class 1 and 203 events of class 0. Of course, not all events occurred at all gauges. If a certain event cannot be detected at a gauge (meaning that no flood occurred), it is assigned with class 

-1, since it is different from all flood classes. With this assumption, a comparison of all gauges and events is made possible. A map of the study area is provided in Figure \ref{Fig:map}.

\begin{figure}
	\centering
	\includegraphics[width=\textwidth]{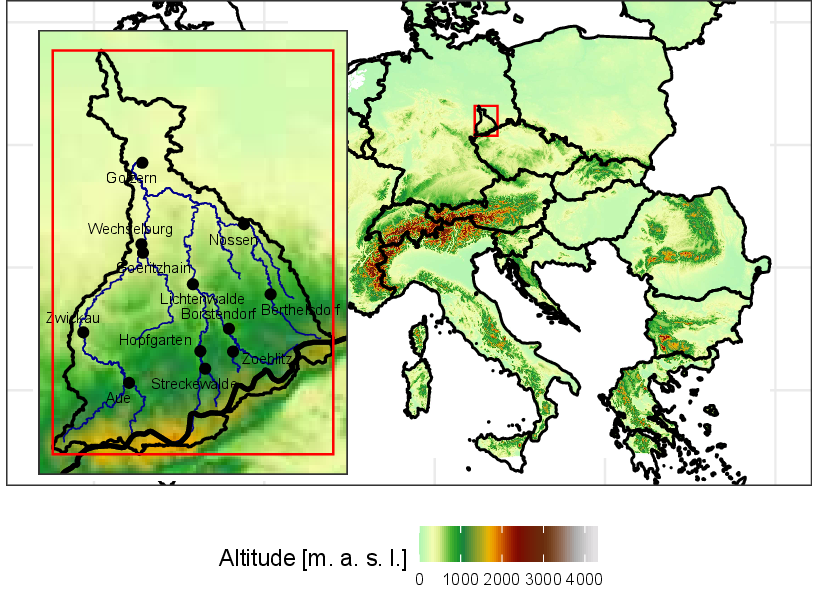}
	\caption{Map of the Mulde River basin in Saxony, Germany and its location in Central Europe.}
	\label{Fig:map}
\end{figure}

\subsection{The Temporal Approach}

Temporal coherence between the flood classes of the gauges is investigated similarly to \cite{fis-sch-sch17}, with shifting ordinal patterns of the flood classes of length $n=4$ by one event per step and compare the patterns between the gauges. Here, the metrical approach is used to include small variations in the patterns. All 13 gauges are considered. The metrical coherence (in percent) for all combination of gauges is given in the supplement. The coherence between the gauges in general is high, indicating strong dependencies between the gauges (Table \ref{TabTime}). Spatial patterns become visible. An example for the time series of flood classes for the four main gauges is given in Figure \ref{Fig:TS}a, where clear coherence especially for extreme events of classes 3 and 4 becomes visible.

The comparison value was estimated using $\hat{q}_N$ in combination with the considered metric. The results are given in the supplement. All comparison values are much smaller than the obtained ordinal pattern dependence. This implies significant dependency between all gauges. 

\vspace{0.3cm}

As second comparison, also two time series $(Z_t)_{t\in \mathbb{Z}}$ of type INGARCH(p,q) \cite{Weiss.2018} have been simulated. More precisely, we simulated
$$Z_t\vert \mathcal{F}_{t-1}\sim\text{Poi}(\nu_t),$$

where $\mathcal{F}_{t-1}$ denotes the history of the process up to time $t-1$ with conditional distribution chosen as Poisson. With this definition, the linear predictor is the conditional mean

$\nu_t=\beta_0+\beta_1 Z_{t-i_1}+\ldots + \beta_p Z_{t-i_p}+\alpha_1\nu_{t-j_1}+\ldots+\alpha_q\nu_{t-j_q}+\eta_1X_{t,1}+\ldots +\eta_rX_{t,r}$.

For the simulation, we chose $\beta_0=2$, $p=1$, $\beta_1=0.3$, $q=0$ and $r=0$ (the INAR(1) model) for sake of simplicity.
For the two time series with $t=1000$, the ordinal pattern metrical coherence has been assessed for 1000 repetitions. In mean, a coherence of 22.2\% has been obtained, which is far below all obtained values for the flood time series (Table \ref{TabTime}). Obviously, even an auto-correlated time series does not lead to such high coherence. If we increase the correlation by assuming $\beta_1=0.6$, the coherence stays at 20.4\%. This can also be seen in Figure \ref{Fig:TS}b.

\begin{figure}
	\centering
	\includegraphics[width=\textwidth]{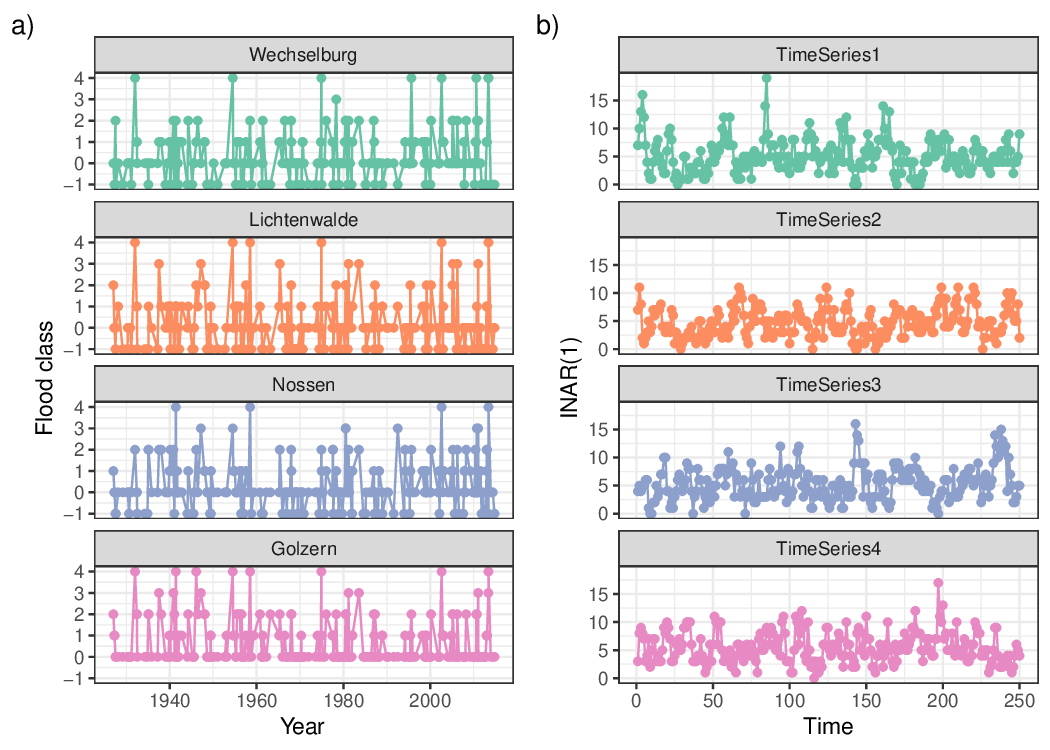}
	\caption{a) Time series  of flood classes of the four main gauges in the Mulde River basin in the years 1927-2015; b) four time series of sample size $n=250$ generated from an INAR(1)-process with $\beta_1=0.6$ without cross-correlation.}
	\label{Fig:TS}
\end{figure}

\begin{table}
	\centering
	\caption{Results of the ordinal pattern dependence for all catchments in the Mulde River basin and simulated INAR(1)-processes. Values are given in \%.}
	\begin{tabular}{l c c c }
		\hline
		Time Series & mean & min & max \\
		\hline
		Mulde flood classes &58.2 & 48.3 & 75.2 \\ 
		INAR(1) with $\beta_1=0.3$ & 18.6 & 26.7 & 22.2\\
		INAR(1) with $\beta_1=0.6$ & 16.5 & 23.9 & 20.4 \\
		\hline
		
	\end{tabular}
	\label{TabTime}
\end{table}

So far, we have only {\it stated} that our method is beneficial to the general handling of ties in the context of ordinal patterns. In the following, we will shortly demonstrate this. For this purpose, two standard handling methods of ties in ordinal patterns are compared with generalized ordinal pattern in terms of ordinal pattern dependence. More precisely, we will use cases b) (randomization of ties) and c) (first appearance) from above, since case a) (the removal from ties) is not meaningful in this context with small patterns and many ties. Again, all 13 gauges from above are compared concerning their ordinal pattern dependence using the metrical approach. To exclude the possibility of dependence of the results on the pattern length, the lengths $n=4$ as well as $n=6$ are used. Since the methodology differs from the one proposed here, the metric has to be adjusted. For these approaches, only even numbers of difference are obtained, hence the metric should only consider these \cite{schnurrdehling2016}. For the case of $n=4$, the metric
\begin{align}\label{weightfkt4}
	w(x):=1\cdot 1_{\{x=0\} } + 0.5 \cdot 1_{\{x=2\} },
\end{align}
is applied, while for $n=6$, we apply

\begin{align}\label{weightfkt3}
	w(x):=1\cdot 1_{\{x=0\} } + 0.75 \cdot 1_{\{x=2\} } + 0.5 \cdot 1_{\{x=4\} } + 0.25 \cdot 1_{\{x=6\} }.
\end{align}

The results are summarized in Table \ref{TabSim}. Detailed results for each combination of gauges are given in the supplement.

\begin{table}
	\centering
	\caption{Results of the ordinal pattern dependence for all 13 gauges with different ties-handling. Values are given in \%.}
	\begin{tabular}{l c c c c c c }
		\hline
		& \multicolumn{3}{c}{$n=4$} & \multicolumn{3}{c}{$n=6$}\\
		Approach & mean & min & max & mean & min & max\\
		\hline
		Generalized Ordinal Pattern Dependence &58 & 48 & 75 & 49 &39 &67\\ 
		Randomized ties & 21 & 13 & 32 & 9.2 & 4.5 & 18\\
		Approach c) & 21 & 7 & 42 & 14 & 7.0 & 30\\
		\hline
		
	\end{tabular}
	\label{TabSim}
\end{table}

The classical ties-handling approaches underestimate the dependence present in the data by altering the ties structure. 
Moreover, the length of the pattern affects the performance of the handling of ties and becomes worse the larger the pattern. This can be explained by the application of the metric. Due to the metrical structure, for small patterns changes in the pattern, as done in the presence of ties, the effect is reduced since these patterns are still similar to other patterns and the metric delivers comparably high values. For larger patterns, this effect decreases.

This application demonstrates the drawbacks of classical ties-handling in ordinal pattern dependence. Due to the changes of the pattern structure, the ordinal pattern dependence is significantly reduced and below the comparison value for significance, hence indicating independence of the time series. With our approach, we overcome these drawbacks in a natural way.

\subsection{The Spatial Approach}

For the spatial approach, we compare the flood classification of the peak discharges at the gauges for each event. Due to computational limits, the maximum of eight gauges (Aue, Zwickau, Wechselburg, Hopfgarten, Zoeblitz, Lichtenwalde, Nossen and Golzern) and the subset consisting only of the estuary Golzern and the three main tributaries Wechselburg, Lichtenwalde and Nossen are considered. This means we obtain 314 ordinal patterns of length $n=8$ respectively $n=4$. For these patterns, a percental frequency is estimated. Of course, it is not only of interest, if there is a coherence between the gauges. This can be assumed for most flood events due to the non-local nature of precipitation. More interestingly, we want to investigate, which pattern occurs most. That is, are there certain gauges that often behave differently than the remaining gauges. For this purpose, we compare the ordinal pattern of each event with all possible theoretical ordinal patterns. The results show, which ordinal pattern occurs most and if this occurrence is significant. The patterns can be assumed to be independent since each pattern belongs to a different flood event which have been separated such that it can be guaranteed that they do not influence each other. Moreover, we have tested the auto-correlation of each time series with the Cram\'er V approach suggested in \cite{Weiss.2018}, where no significant auto-correlation occurred for any of the gauges (with mean values between 0.1 and 0.16 for lags up to $k=100$).

If we only consider the four estuarial catchments Wechselburg, Lichtenwalde, Nossen and Golzern, representing the three tributaries and the estuary, we can observe that the unit pattern $(1,1,1,1)=e_4$ representing the vectors $k\cdot e_4$ $(k=-1,\ldots,4)$ occurs with highest frequency of all patterns, being significantly higher than the theoretical value (Table \ref{TabSpatial}). This implies that in general we have very similar flood classes for events in the tributaries and the estuary and that their occurrence is significantly higher than expected. The second highest frequency can be observed for the pattern where the central tributary has flood class that differs by one from the remaining classes and third highest where the Golzern gauge and exactly one of the tributaries differ by one class from the remaining two gauges. Obviously, the magnitude of a flood at the estuary is determined by the largest flood class of the three tributaries. The fourth most type of pattern seems to be surprising on the first sight: it is the pattern where the estuary gauge has a flood class differing by one class from the remaining gauges. From the hydrological perspective this can be explained by the phenom of overlaid flood waves. Due to temporally similar occurrence of flood event in the tributaries caused by spatially extended rainfall, often flood waves can overlap at the estuary leading to so called superpositioned waves \cite{Bloeschl.2013}. This leads to much larger flood waves in the estuary and hence to much larger flood classes. In theory, it should be a significantly lower coherence and hence this pattern is significant for flood classes.
Non-observed patterns ($0\%$) are those where two gauges show flood classes that differ by more than two classes from the other tributaries. Obviously, the most extreme events indicated by large flood classes do not occur locally in one tributary. Instead, at least moderate floods can be observed for the whole basin. We would expect these patterns to occur in $1\%$ of all patterns. Also non-observed are patterns where only one gauge has the flood class zero ($0\%$). This emphasizes the strong dependence within the basin, leading to similar reactions for the whole basin.

\begin{table}
	\centering
	\caption{Frequency of the most-frequent ordinal patterns of flood classes for the four main gauges in the Mulde River basin and theoretical comparison values. Vectors of flood classes refer to the order (Golzern, Wechselburg, Lichtenwalde, Nossen), $k$ can take integer values between -1 and 4, respectively 0 and 4 for Golzern gauge. Values are given in \%.}
	\begin{tabular}{l l c c  }
		\hline
		Vector of flood classes & Ordinal pattern & observed & theoretical \\
		\hline
		(k,k,k,k) & (1,1,1,1) &58.3 & 48.2 \\ 
		(k,k,k+1,k) & (1,1,2,1) & 21.1 & 13.2  \\
		(k+1,k+1,k,k) & (2,2,1,1) & \multirow{3}{*}{$\left\rbrace\begin{array}{l}
				\\
				8.1\\
				\\
			\end{array}\right.$} & \multirow{3}{*}{3.9}\\
		(k+1,k,k+1,k) & (2,1,2,1) & & \\
		(k+1,k,k,k+1) & (2,1,1,2) & &\\
		(k+1,k,k,k) & (2,1,1,1) &6.4 & 3.5\\
		\hline
		
	\end{tabular}
	\label{TabSpatial}
\end{table}

If we include all gauges, i.e., also headgauges, the results remain similar. Due to the increased possibilities of patterns, the overall proportion of distinct patterns reduces, but still the same significant patterns can be observed. For example, the unit pattern has a frequency of $23.2\%$, where we would expect $7.3\%$. Nevertheless, also new structures within the basin become visible. Interestingly, patterns where only a sub-basin has a flood class one larger than the remaining sub-basins, which is proceeding to the main gauges and the outlet gauge, has highest frequency of $40.16\%$. This implies that the head gauges do not have such that a high dependence as the gauges located further downstream. Weather patterns do not affect those gauges alike, but instead more frequent affect only the downstream parts, further away from the Ore mountains. 

\section{Summary and Conclusions}
The detection of coherence between time series in large data sets is an important and challenging task. There exist various methods to do this for different settings. Here, we have proposed a method based on ordinal patterns. Ordinal patterns offer a simple consideration where only the relative positions of subsequent data points are considered and not the values itself. They have been applied in many disciplines like medicine, finance and hydrology. However, what has been overlooked so far was the handling of ties. Ties occur in many time series and are of particular importance when analyzing categorical data. In this work, we have explicitly considered ties and hence: generalized patterns. New limit theorems under short range dependent conditions were developed. The general concept of short range dependence (absolute regularity) that was used makes an application to many classical time series models possible. The general structure of the proofs using U-Statistics allows for an extension to further dependence concepts.

The methodology then was applied to flood classes in a river basin in Germany. It was shown that there exists a highly significant dependence between the flood classes of neighboring catchments. At the same time, certain spatial patterns where detected that occurred significantly frequent. These patterns were linked to flood-generating processes like the Vb weather pattern, which is typical for this region and often causes extreme floods, and flood superposition.
Finally, a simulation study was performed that demonstrates how the proposed handling of ties is more efficient than the classical treatment of ties (randomization, removal) and does not lead to a loss of information.

Future extensions of ordinal patterns for data with ties could be the consideration of different dependence structures like near-epoch dependence, which would allow an application to more classes of time series models such as GARCH. Moreover, long-range dependent processes could be considered which also may be of interest in the case of hydrological data.

\section*{Acknowledgments}

We would like to thank the Saxon State Office for Environment, Agriculture and Geology for providing the discharge data, which are available from the data portal iDA (www.umwelt.sachsen.de/umwelt/46037).

Funding: This work was supported by the German Research Foundation (DFG) [grant numbers FOR 2416; SCHN 1231-3/2].

\bigskip
\begin{center}
	{\large\bf SUPPLEMENTARY MATERIAL}
\end{center}

\begin{description}
	
	\item[Table S1: Generalized Ordinal Pattern Dependence for Temporal Approach with n=4:] Table with all ordinal pattern dependencies (in percent) estimated with the generalized ordinal pattern approach with pattern length $n=4$. Please note that the results are symmetric.  (.csv file)
	
	\item[Table S2: Generalized Ordinal Pattern Dependence for Temporal Approach with n=6:] Table with all ordinal pattern dependencies (in percent) estimated with the generalized ordinal pattern approach with pattern length $n=6$. Please note that the results are symmetric.  (.csv file)
	
	\item[Table S3: Comp. value of Gen. Ordinal Pattern Dependence for Temporal Approach with n=4:] Table with all comparison values for ordinal pattern dependencies (in percent) estimated with the generalized ordinal pattern approach with pattern length $n=4$. Please note that the results are symmetric.  (.csv file)
	
	\item[Table S4: Comp. value of Gen. Ordinal Pattern Dependence for Temporal Approach with n=6:] Table with all comparison values for ordinal pattern dependencies (in percent) estimated with the generalized ordinal pattern approach with pattern length $n=6$. Please note that the results are symmetric.  (.csv file)
	
	\item[Table S5: OPD with Ties handled by Randomization for Temporal Approach with n=4:] Table with ordinal pattern dependencies (in percent) estimated with the classical ordinal pattern approach with pattern length $n=4$ and ties handled by randomization (case b). Please note that the results are symmetric.  (.csv file)
	
	\item[Table S6: OPD with Ties handled by Randomization for Temporal Approach with n=6:] Table with ordinal pattern dependencies (in percent) estimated with the classical ordinal pattern approach with pattern length $n=6$ and ties handled by randomization (case b). Please note that the results are symmetric.  (.csv file)
	
	\item[Table S7: OPD with Ties handled by Case c for Temporal Approach with n=4:] Table with ordinal pattern dependencies (in percent) estimated with the classical ordinal pattern approach with pattern length $n=4$ and ties handled by case c. Please note that the results are symmetric.  (.csv file)
	
	\item[Table S8: OPD with Ties handled by Case c for Temporal Approach with n=6:] Table with ordinal pattern dependencies (in percent) estimated with the classical ordinal pattern approach with pattern length $n=6$ and ties handled by case c. Please note that the results are symmetric.  (.csv file)
	
\end{description}

\bibliography{References}

\end{document}